\documentclass[11pt,reqno]{amsart} 
\usepackage[T1]{fontenc}
\usepackage[utf8]{inputenc}
\usepackage{amsmath,amssymb,amsthm,mathtools}
\usepackage[colorlinks=true,linkcolor=blue,citecolor=blue,urlcolor=blue]{hyperref}
\usepackage{booktabs}
\usepackage{graphicx}
\numberwithin{equation}{section}
\newtheorem{theorem}{Theorem}[section]

\theoremstyle{definition}

\theoremstyle{remark}

\makeatletter
\renewcommand\theequation{\arabic{equation}}
\makeatother

\newcommand{\pr}{\mbox{Pr}}
\usepackage{xcolor}  

\title{Energetic Protection, Monotonicity and Switching Far from Equilibrium}
\author{\NoCaseChange{Uğur Çetiner}}
\begin{document}
\maketitle

\begin{center}
\footnotesize
Department of Physics, Koç University\\
Rumelifeneri Yolu 34450 Sarıyer, İstanbul, Türkiye\\
\texttt{ucetiner@ku.edu.tr}
\end{center}
\vspace{6pt}
\begin{abstract}
At equilibrium, the ratio of two steady-state probabilities is a Boltzmann factor, set by a free-energy difference. Such ratios are the natural, normalization-independent readouts of both thermodynamics and information processing, and are often used as a measure of fidelity in biophysical systems. What becomes of these ratios once a system is driven from equilibrium, where the Boltzmann factor no longer holds? Representing Markov processes as graphs and their steady states as averages over a distribution on spanning trees, the \emph{arboreal distribution}, we track the ratio $\pi_i/\pi_j$ under driving along \emph{energetic edges}, where detailed balance is broken, relative to its equilibrium value. Our central finding is that a chosen ratio can stay exactly locked to its equilibrium value arbitrarily far from equilibrium, a phenomenon we call \emph{energetic protection}, whenever an algebraic equality between spanning-tree weights holds. Just as detailed balance constrains rates around cycles, energetic protection constrains weights across trees, providing a new mechanism for robustness against fluctuations in the driving force, such as variations in ATP concentration. Away from this equality, single-edge driving collapses the response onto two arboreal coefficients and forces it to be monotonic, so nonmonotonic single-edge control is impossible at any strength. With two energetic edges, protection and monotonicity combine into a \emph{thermodynamic switch} that holds a function at its equilibrium value for as long as desired and releases it sharply. Equilibrium is known for the limits it places on information processing. We show that new constraints, both no-go principles and exact invariances, survive far from equilibrium. These results reveal how the localization of energy expenditure governs the functional logic of nonequilibrium systems in physics and biology.
\end{abstract}

\section*{Introduction}
Unlike classical thermodynamics, the thermodynamics of information processing remains a vibrant and rapidly developing field, continuing to reveal results that are both conceptually surprising and important. For instance, while a system at thermodynamic equilibrium cannot perform work or sustain directed motion \cite{julicher1997modeling}, computation itself need not intrinsically require dissipation. As shown by Bennett, Toffoli and others, the unavoidable thermodynamic cost is tied not to information processing per se, but to logical irreversibility \cite{landauer1961irreversibility}: in principle, logically reversible computation can be carried out with arbitrarily little dissipation \cite{bennett1973logical,toffoli1980reversible_tm151,feynman2018feynman}. This insight stands as one of the major conceptual advances of twentieth-century physics. 

We also know that, in the absence of energy expenditure, thermodynamic equilibrium imposes nontrivial tradeoffs and limits on certain functional information-processing capabilities. A celebrated example is Hopfield's kinetic proofreading mechanism, which showed that in molecular recognition tasks such as DNA replication and protein synthesis, thermodynamic equilibrium imposes a fundamental bound on error correction \cite{hopfield1974}. Put differently, the capacity of a system to discriminate between correct and incorrect substrates is intrinsically limited at thermodynamic equilibrium. More recently, it has been discovered that thermodynamic equilibrium imposes fundamental limits not only on error correction but also on other functional aspects of information processing, such as the sharpness of input-output relations in Markovian systems \cite{martinez2024hill}. 

In this paper, we push this line of inquiry beyond the equilibrium setting. We show that even when a system operates arbitrarily far from thermodynamic equilibrium, certain information-processing capabilities remain subject to previously unrecognized constraints, both no-go theorems that rule out certain forms of control and exact invariances that hold a chosen function fixed under driving. Our analysis is formulated for continuous-time Markov processes on discrete state spaces. Assuming the existence of a unique steady state, we express this state through a graph-theoretic framework that we have recently introduced \cite{ccetiner2022reformulating}. In this representation, the states of the Markov process correspond to the vertices of a graph, allowed transitions correspond to edges, and transition rates appear as edge labels. This correspondence allows us to treat the underlying graph as an independent mathematical object in its own right and to use its structure to gain algebraic access to steady states. It is precisely this perspective that enables us to prove theorems that hold universally across all finite-state Markov processes \cite{cetiner2023universal}.

Steady states come in two sharply distinct forms: those consistent with thermodynamic equilibrium, which we denote by \(\pi^{\mathrm{eq}}\), and those corresponding to nonequilibrium steady states, which we denote by \(\pi\). Thermodynamic equilibrium admits a precise characterization, manifesting itself as an algebraic constraint on the transition rates around every cycle of the graph (see Setup for a precise statement). The central observable in our analysis is the ratio of steady-state probabilities, $\pi_i/\pi_j,$ which compares the likelihood of occupying state \(i\) to that of occupying state \(j\) in the steady state.

Ratios of this kind arise naturally as measures of functional performance throughout nonequilibrium biophysics, and in each example below the quantity compared is the ratio of steady-state occupancies of two molecular states. For example, in chromatin-mediated gene activation, transcriptional specificity is quantified by the discrimination ratio $R = [AT]/[DT] = \pi_{AT}/\pi_{DT}$, the ratio of steady-state occupancies of correctly and incorrectly marked promoters in their acetylated ($AT$) and deacetylated ($DT$) accessible states, an ATP-driven remodeling step can amplify this ratio far beyond its equilibrium value \cite{blossey2008kinetic}. A similar ratio appears in protein translation fidelity, where the error is defined as the rate of incorrect relative to correct amino acid incorporation into the growing peptide chain, each rate being proportional to the steady-state probability that a certain molecular state is occupied \cite{hopfield1974,cetiner2023universal}. In protein quality control, $\pi_N$ and $\pi_M$ are the steady-state probabilities of the native (N) and misfolded states (M), and ATP-consuming chaperones drive the ratio $\pi_N/\pi_M$ above its equilibrium value $\pi_N^{\mathrm{eq}}/\pi_M^{\mathrm{eq}}$ \cite{xu2018cochaperones}. Therefore, if \(i\) is interpreted as an undesired state and \(j\) as a desired one, then \(\pi_i/\pi_j\) serves as a natural measure of error, with smaller values corresponding to improved performance. In this context, we define the normalized response as
\[
K_{ij}\equiv \frac{\pi_i/\pi_j}{\pi_i^{\mathrm{eq}}/\pi_j^{\mathrm{eq}}},
\]
which measures how far nonequilibrium driving enhances or suppresses performance relative to equilibrium.

We defer the proofs to the Results section and summarize here the three main findings of this work. First, consider a Markov process initially at thermodynamic equilibrium. Without loss of generality, we drive the system out of equilibrium by multiplying a single transition rate by a factor \(m>1\). For example, suppose the equilibrium rate of the transition \(z_1\to z_2\), denoted by \(\ell_{\mathrm{eq}}(z_1\to z_2)\), is changed to $\ell(z_1\to z_2)=m\,\ell_{\mathrm{eq}}(z_1\to z_2)$. One may think of this perturbation as arising from coupling the transition \(z_1\to z_2\) to an energetically favorable reaction such as ATP hydrolysis. More specifically, $m=e^{\beta\Delta\mu}$ where $\beta=1/k_BT$ and $\Delta\mu$ is the free energy supplied by the coupled reaction (for ATP hydrolysis, $\Delta\mu=\mu_{\mathrm{ATP}}-\mu_{\mathrm{ADP}}-\mu_{\mathrm{P_i}}$). Accordingly, $\ln m$ can be considered as the thermodynamic driving force. Although only a single transition rate is modified, the system is thereby driven out of equilibrium and settles into a nonequilibrium steady state. Since the edge \(z_1\to z_2\) is the one along which thermodynamic equilibrium is broken, we refer to it as the \emph{energetic edge}. The response then becomes a function of \(m\), which we denote by \(K_{ij}(m)\). We will show that \(K_{ij}(m)\) is always monotonic in \(m\). In other words, when an equilibrium system is driven through a single energetic edge, the response cannot be nonmonotonic. This reveals a new limitation on the information-processing capacity of nonequilibrium steady states. Remarkably, the entire nonequilibrium complexity is compressed into just two numbers, the \emph{arboreal coefficients} \(A_+(i)\) and \(A_+(j)\) (Eq.~\eqref{e-acoeffs}). Their ordering alone determines the fate of the response: \(A_+(i)>A_+(j)\) implies that \(K_{ij}(m)\) grows monotonically with \(m\), while \(A_+(i)<A_+(j)\) implies that it decays monotonically.

Our first main finding is a new nonequilibrium phenomenon that we call \emph{energetic ratio protection} or simply \emph{energetic protection}. It occurs if and only if $A_+(i)=A_+(j)$, a condition that translates into algebraic constraints on the transition rates, expressed through the weights of certain subgraphs called \emph{spanning trees}: cycle-free subgraphs that include every vertex and are directed so that each vertex has a unique path to a chosen reference vertex. When the condition holds, $K_{ij}(m)=1$, or equivalently $\pi_i/\pi_j=\pi_i^{\mathrm{eq}}/\pi_j^{\mathrm{eq}}$, for all values of $m$: the ratio stays pinned to its equilibrium value however strongly the system is driven. This invariance under energetic driving is what motivates the name. Intracellular ATP levels can shift substantially with metabolic state, stress, and development, and vary considerably from cell to cell \cite{ataullakhanov2002determines,de2014dynamics,ozalp2010time}. Because a protected ratio is fixed at its equilibrium value independently of the driving strength, energetic protection provides a mechanism by which a chosen observable can remain calibrated to a fixed baseline even as the ATP-dependent drive fluctuates. Dissipation still occurs and the system may sit arbitrarily far from thermodynamic equilibrium, yet the drive leaves the protected ratio unchanged, its effect exactly cancelled by an algebraic balance among certain spanning-tree weights determined by the arboreal coefficients.

As an illustration of energetic protection, consider the four-state Markov process in Fig.~\ref{fig:Fig1}A, and suppose initial thermodynamic equilibrium is broken by modifying the rate \(1\to 3\), which is highlighted in red. The spanning trees rooted at vertex \(1\) are also shown in Fig.~\ref{fig:Fig1}B. Our analysis shows that the response is determined solely by the weights of two spanning trees, with the weight of each spanning tree defined as the product of its edge labels. If the weight of the seventh spanning tree exceeds that of the sixth, \(w(T_{\#7})>w(T_{\#6})\), then \(A_+(2)>A_+(4)\), and the response is monotonically increasing, as shown in Fig.~\ref{fig:Fig1}C. If instead \(w(T_{\#7})<w(T_{\#6})\), then \(A_+(2)<A_+(4)\), and the response is monotonically decreasing, as shown in Fig.~\ref{fig:Fig1}E. Finally, if \(w(T_{\#6})=w(T_{\#7})\), then \(A_+(2)=A_+(4)\), so that \(K_{24}(m)=1\) for all \(m\), as shown in Fig.~\ref{fig:Fig1}D. In other words, the ratio \(\pi_2/\pi_4\) is protected against energetic perturbations whenever the constraint \(w(T_{\#6})=w(T_{\#7})\) holds. This shows a delicate connection between information-processing capabilities and the combinatorial structure of the graph, which would have been difficult to reveal without our formulation of steady states in Eq.~\eqref{e-rhoness}. Crucially, this protection requires no further fine-tuning: as long as the system is initially at thermodynamic equilibrium and the single equality \(A_+(i)=A_+(j)\) holds---which, for the example in Fig.~\ref{fig:Fig1}A, reduces to \(w(T_{\#6})=w(T_{\#7})\)---the response remains pinned to its equilibrium value for all driving strengths \(m\), independently of the values of the remaining transition rates.

Our second result shows that, with two energetic edges, one can construct a
\emph{thermodynamic switch} in which an initially protected response is held
exactly at its equilibrium value for as long as desired and then released, by a
second drive, to a substantially different level. Let the second energetic edge
be \(z_3\to z_4\), with \(\ell(z_3\to z_4)=n\,\ell_{\mathrm{eq}}(z_3\to z_4)\) so that the response then depends on two driving parameters, \(K_{ij}(m,n)\). As we will show, even in the presence of two energetic edges, holding one parameter fixed while varying the other leaves the response monotonic in the varying parameter. Combining this monotonicity with energetic protection generates switching behaviour. To illustrate, consider again the four-state Markov process in Fig.~\ref{fig:Fig1}A driven by two energetic edges. Starting from thermodynamic equilibrium, we impose the constraint \(w(T_{\#6})=w(T_{\#7})\), which protects the response \(K_{24}\) against energetic perturbations along the edge \(1\to 3\) with the driving parameter \(m\), and then introduce a second energetic edge, \(2\to 3\), with the driving parameter \(n\). The protocol proceeds in two stages. In the first stage, we fix \(n=1\) and increase \(m\) up to \(m=100\). Throughout this stage \(\pi_2/\pi_4\) remains pinned to its equilibrium value, exactly as required by energetic protection, even though the system is driven away from thermodynamic equilibrium. In the second stage, we hold \(m=100\) fixed and increase \(n\). For a suitable choice of parameters (given in Table~\ref{t-2}) the response then decreases monotonically from its equilibrium value down to a much lower level, as shown in Fig.~\ref{fig:Fig2}, where a nearly 10 fold decrease is achieved only by a small shift in the cycle affinity or driving force ($\Delta \ln (n) \approx 2.3$). In this way the ratio \(\pi_2/\pi_4\) can be held exactly at its equilibrium value for as long as desired and then swept to a much lower value by activating the second drive.

Our third and final result concerns nonmonotonicity. Driving a single energetic edge yields a monotonic response, and the same holds when two energetic edges are driven independently. A nonmonotonic response therefore requires the two energetic edges to be \emph{coupled}, for example through a relation
\(n=n(m)\), as occurs when both edges are powered by the same free-energy source such as ATP hydrolysis \cite{hill2012free}. In the representative case \(n=m\), we derive conditions for the emergence of nonmonotonicity (Fig.~\ref{fig:Fig4}). In this case the response is a ratio of quadratics in \(m\), so whenever it is nonmonotonic it exhibits only a single bump, i.e.\ a single interior minimum or maximum. More broadly, these results show how the structure and coupling of nonequilibrium drives can be used to design and control functional behaviour in systems operating far from thermodynamic equilibrium.
\begin{figure}[t]
  \centering
  \includegraphics[width=\linewidth]{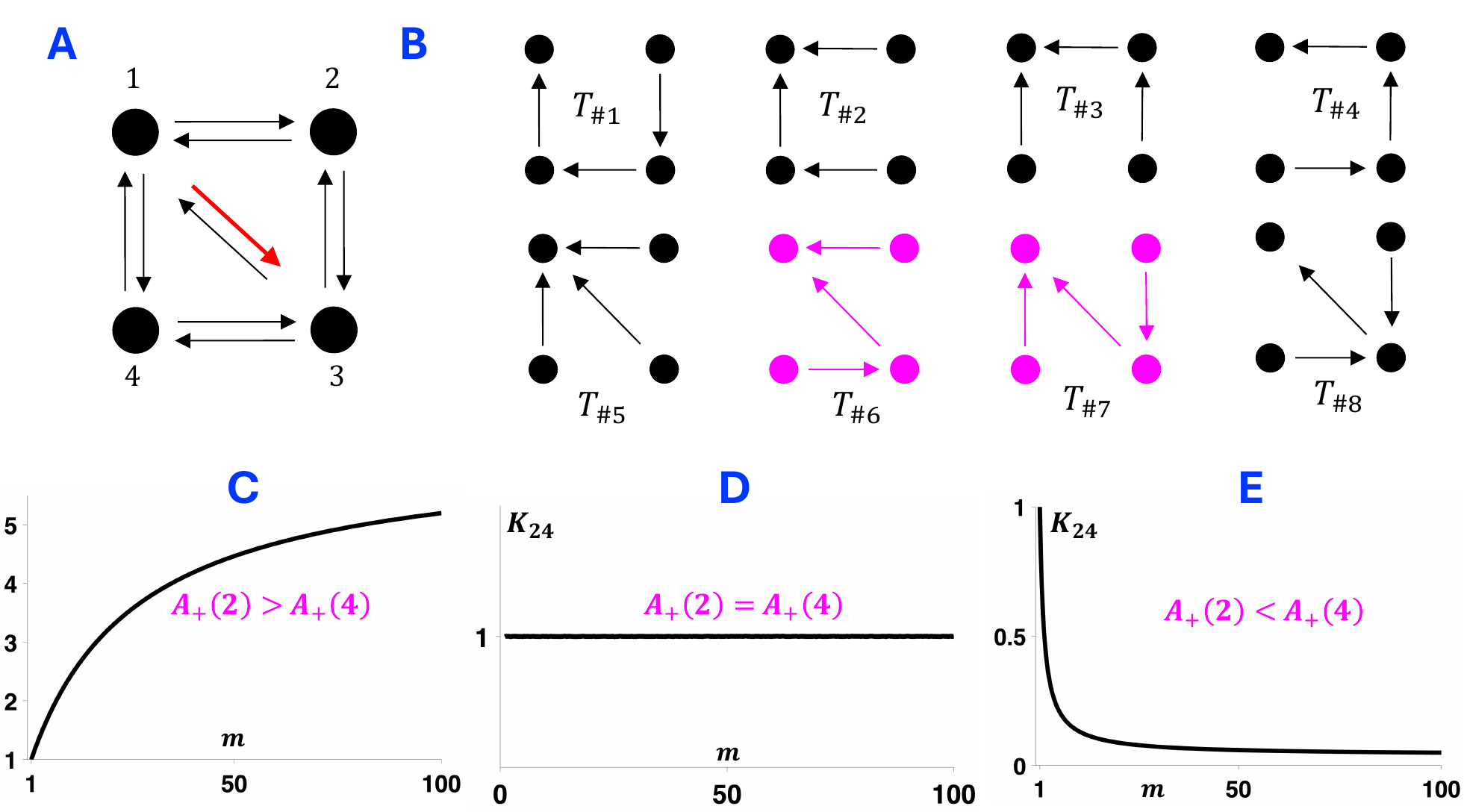}
\caption{\textbf{Single-edge driving yields either a monotonic response
or energetic protection.}
(\textbf{A})~A four-state Markov process is driven out of equilibrium by scaling
the single transition $1\to3$ as
$\ell(1\to3)=m\,\ell_{\mathrm{eq}}(1\to3)$ (red), with every other rate
held at its equilibrium value.
(\textbf{B})~The eight spanning trees rooted at vertex~$1$; the two trees
$T_{\#6}$ and $T_{\#7}$ (magenta) compete to set the response.
(\textbf{C--E})~The normalized response $K_{24}(m)\equiv(\pi_2/\pi_4)/(\pi_2^{\mathrm{eq}}/\pi_4^{\mathrm{eq}})$ is fixed by the relative weights of these two trees: it increases monotonically when $w(T_{\#7})>w(T_{\#6})$~(\textbf{C}), decreases monotonically when $w(T_{\#7})<w(T_{\#6})$~(\textbf{E}), and stays pinned at $K_{24}(m)=1$ for all~$m$ when $w(T_{\#7})=w(T_{\#6})$~(\textbf{D}), so that the ratio $\pi_2/\pi_4$
is \emph{energetically protected} against the perturbation.}
  \label{fig:Fig1}
\end{figure}
\section*{Setup and preliminary Results}
We model the system as a continuous-time, finite-state Markov process $X(t)$. In the \emph{linear framework} \cite{ccetiner2022reformulating,gunawardena2012linear}, the process is represented by a directed graph $G=(V_G,E_G)$ with labeled edges and no self-loops: $V_G=\{1,2,\dots,n\}$ is the set of states, and $E_G\subseteq \{(i,j)\in V_G\times V_G:\ i\neq j\}$ is the set of allowed transitions. Each directed edge $(i,j)\in E_G$, written $i\to j$, carries a label $\ell(i\to j)$, interpreted as the transition rate from state $i$ to state $j$ (with units $\mathrm{time}^{-1}$),
\begin{equation}
\ell(i \to j)\;=\;\lim_{\Delta t \to 0}\;
\frac{\Pr\!\big\{X(t+\Delta t)=j \,\big|\, X(t)=i\big\}}{\Delta t},
\label{e-rates}
\end{equation}
whenever the limit exists and is nonzero.

Let $p_i(t)=\Pr\{X(t)=i\}$ and $p(t)=(p_1(t),\dots,p_n(t))^\top$ be an $n\times 1$ column vector. The master equation describes how the probabilities evolve in time as follows,
\begin{equation}
\frac{d\,p(t)}{dt}\;=\; {\mathcal L}(G)\,p(t),
\label{e-master}
\end{equation}
where the graph Laplacian ${\mathcal L}(G)$ (the infinitesimal generator) is
\begin{equation}\label{e-Lmatrix}
{\mathcal L}(G)_{ji} =
\begin{cases}
\ell(i \to j) & \text{if $i \neq j$},\\[2pt]
-\sum_{k \neq i}{\ell(i \to k)} & \text{if $i=j$}.
\end{cases}
\end{equation}
We assume \emph{bidirectionality} (if $i\to j$ then $j\to i$) and that $G$ is strongly connected, meaning that every vertex can be reached from every other vertex. Then there is a unique steady state $\pi$ satisfying ${\mathcal L}(G)\pi=0$ (equivalently, $\pi\in\ker{\mathcal L}(G)$). In addition, if
\begin{equation}
    \pi_i\,\ell(i\to j)=\pi_j\,\ell(j\to i),\qquad \forall\, i,j\in V_G,
\end{equation}
then $\pi$ is an \emph{equilibrium steady state} (detailed balance holds). Otherwise, the dynamics relax to a \emph{nonequilibrium steady state}. In this case, time-reversal symmetry or detailed balance is broken and maintaining $\pi$ typically requires continuous energy expenditure.

A trajectory $\{X(t)\}_{t\ge 0}$ induces a directed path
$P=[i_1,i_2,\dots,i_k]$ on $G$ (written $i_1\to i_2\to \cdots \to i_k$).
A path is \emph{simple} if all vertices are distinct; it is \emph{closed} if $i_1=i_k$; a \emph{simple cycle} is a closed path with no repeated vertices except the start/end.
For a directed path $P:i_1\to \cdots \to i_k$, define the \emph{path action} as follows \cite{lebowitz1999gallavotti},
\begin{equation}
S(P)
=\sum_{r=1}^{k-1}\ln\frac{\ell(i_r\to i_{r+1})}{\ell(i_{r+1}\to i_r)}.
\label{e-spp}
\end{equation}
For a cycle, $C$, $S(C)$ is the cycle affinity.

A subgraph $H=(V_H,E_H)$ satisfies $V_H\subseteq V_G$ and $E_H\subseteq E_G$ and inherits edges/labels from $G$. Viewing $G$ as undirected, a \emph{tree} is a connected and acyclic subgraph; a \emph{spanning tree} contains all vertices.
A \emph{rooted spanning tree} rooted at $r$ is a spanning subgraph in which every vertex has a unique directed path to $r$ (equivalently, every non-root vertex has exactly one outgoing edge).
For any subgraph $H$, its weight is defined as
\[
w(H)=\prod_{\,i\to j\in E_H}\ell(i\to j),
\]
which extends naturally to a set $A$ of subgraphs by $w(A)=\sum_{H\in A} w(H)$.

\subsection*{Obtaining equilibrium vs nonequilibrium steady states}
At thermodynamic equilibrium, every cycle has zero affinity ($S(C)=0$ for all cycles $C$), which implies \emph{path-independence}: for any two simple paths $P,Q$ from $i$ to $j$, $S(P)=S(Q)$ \cite{ccetiner2022reformulating}. Thus we may define
\[
S_{\mathrm{eq}}(i,j)\equiv S(P)\qquad \text{for any } P \in M(i,j),
\]
where $M(i,j)$ is the set of simple paths from $i$ to $j$. To obtain the equilibrium steady-state solution, we employ the following protocol. Fix a reference vertex $r$ (arbitrarily) and define
\begin{equation}\label{e-mu}
\mu_i(G)=e^{-S_{\mathrm{eq}}(i,r)},
\end{equation}
so that $\mu(G)\in\ker\mathcal{L}(G)$. Since $\ker\mathcal{L}(G)$ is one-dimensional, $\pi_i\propto \mu_i(G)$, and normalization gives
\begin{equation}\label{e-normalized}
\pi_i=\frac{e^{-S_{\mathrm{eq}}(i,r)}}{\sum_j e^{-S_{\mathrm{eq}}(j,r)}}\,.
\end{equation}

Away from equilibrium, $S$ is path-dependent, so a unique solution can no longer be obtained from a single path action, as in Eq.~\eqref{e-normalized}. Recently, however, we found that the steady-state solution can still be constructed by averaging these path-dependent contributions over a probability distribution on spanning trees rooted at the reference vertex, which we call the \emph{arboreal distribution}.

Let $\Theta_i(G)$ be the set of spanning trees rooted at $i$. Define the probability of $T \in \Theta_r(G)$ by
\begin{equation}\label{arboreal}
\pr_{\Theta_r(G)}(T)=\frac{w(T)}{\sum_{T\in\Theta_r(G)} w(T)}=\frac{w(T)}{w(\Theta_r(G))}\,.
\end{equation}
For each $T\in\Theta_r(G)$, let $T_i\in M(i,r)$ be the unique simple path from $i$ to the root $r$ within $T$. Then a canonical basis element $\rho(G)\in\ker\mathcal{L}(G)$ is given by the following expression \cite{ccetiner2022reformulating},
\begin{equation}\label{e-rhoness}
\rho_i(G)=\sum_{T\in\Theta_r(G)}\pr_{\Theta_r(G)}(T)\,e^{-S(T_i)}
\equiv \left\langle e^{-S(T_i)}\right\rangle,
\end{equation}
where the angular brackets represent an average over the arboreal distribution and the steady-state solution is obtained by normalization:
\begin{equation}\label{e-NessFinal}
\pi_i=\frac{\left\langle e^{-S(T_i)}\right\rangle}{\sum_j \left\langle e^{-S(T_j)}\right\rangle}\,.
\end{equation}
At equilibrium, $S(T_i)=S_{\mathrm{eq}}(i,r)$ for all $T$, so Eqs.~\eqref{e-rhoness} and \eqref{e-NessFinal} reduce to Eqs.~\eqref{e-mu} and \eqref{e-normalized}, respectively. Therefore, Eq.~\eqref{e-rhoness} extends equilibrium formalism (Eq.~\eqref{e-mu}) to nonequilibrium steady states.
\subsection*{Matrix-tree theorem vs Eq.~\eqref{e-rhoness}} 
The matrix-tree theorem (MTT) provides another basis element \(\tilde{\rho}(G) \in \ker \mathcal{L}(G)\) through
\begin{equation}\label{e-rhomtt}
\tilde{\rho}_i(G)=\sum_{T\in\Theta_i(G)} w(T).
\end{equation}
The steady-state solution is then obtained by normalization,
\begin{equation}\label{e-mtt}
\pi_i=\frac{\tilde{\rho}_i(G)}{\sum_j \tilde{\rho}_j(G)}.
\end{equation}

MTT is at the intersection of several graph-theoretic studies in mathematics and physics. A version of it can be dated to Kirchhoff's work on electrical networks, while the version used here for directed graphs is usually attributed to Tutte \cite{gunawardena2012linear,tutte1948dissection}. Equivalent formulas were later rediscovered in biology and nonequilibrium physics, for instance in the work of King and Altman \cite{king1956schematic} and Hill \cite{hill2012free}, where graphs became a natural language for mesoscopic stochastic models of enzymes, transporters, pumps, and other molecular machines. Although elegant and exact, the MTT has two important drawbacks. First, consider graphs built from \(k\) binary sites, where each mesostate specifies whether a given feature is present or absent at each site, such as whether a binding molecule is bound or unbound, giving \(2^k\) mesostates in total. Even in this very structured setting, the number of spanning trees grows extraordinarily fast: it is 4 for \(k=2\), 384 for \(k=3\), and 42,467,328 for \(k=4\). Because the exact steady-state probabilities require summing contributions from all of these trees even when there is a single energetic edge, the matrix-tree theorem becomes overwhelmingly complicated after only a modest increase in system size (algebraic explosion). Second, Eq.~\eqref{e-mtt} expresses the steady state as a rational function of edge labels with a highly nontransparent monomial structure: both numerator and denominator are sums of monomials, each associated with a spanning-tree weight. In this form, the steady state is exact but thermodynamically obscure, because the expression is organized by combinatorics rather than by physically meaningful quantities such as path actions.

Eq.~\eqref{e-rhoness} may be viewed as a reformulation of Eq.~\eqref{e-mtt} that addresses both issues. On the conceptual side, it rewrites the steady state in a form that is physically interpretable: paths are assigned actions and then averaged with respect to the arboreal distribution. In this way, the nonequilibrium steady state acquires a structure that is much closer to the familiar logic of physics. On the structural side, the relevant path-action values need not grow with graph size. Indeed, as we show below, in the presence of a single energetic edge, \(S(P)\) can take at most three distinct values, independent of the size of the graph. Our formulation therefore exposes a remarkable simplification hidden inside nonequilibrium steady states: while the underlying spanning-tree expansion may be combinatorially enormous, its thermodynamic content can collapse to a small set of action sectors.
\section*{Results}
\subsection*{Path action when there is a single energetic edge} 
Let us perturb an equilibrium Markov process with labels $\ell_{\mathrm{eq}}(i\to j)$ by modifying a single transition rate,
\[
\ell(z_1\to z_2)\mapsto m\,\ell_{\mathrm{eq}}(z_1\to z_2),
\]
where $z_1\to z_2$ is the \emph{energetic edge}, since it is the edge at which detailed balance, and hence time-reversal symmetry, is broken. After this perturbation, the dynamics relax to a unique nonequilibrium steady state and the action $S(P)$ becomes \emph{path-dependent}. Nevertheless, by Eq.~\eqref{e-spp}, for any simple directed path $P:i\to\cdots\to j$, the deviation from the equilibrium action can take only three possible values, determined by how the energetic edge is traversed along the path. Writing $E_P$ for the directed edges used by $P$, define
\begin{equation}\label{e-path_entropies}
\sigma_{kl}(P)=
\begin{cases}
1, & z_k\to z_l \in E_P,\\[2pt]
-1, & z_l\to z_k \in  E_P,\\[2pt]
0, & z_k\to z_l \notin  E_P,\ z_l\to z_k \notin E_P,
\end{cases}
\end{equation}
then
\begin{equation}\label{e-path_entropies2}
S(P)=S_{\mathrm{eq}}(i,j)+\sigma_{12}(P)\ln(m).
\end{equation}
Thus, regardless of the size of the graph, $S(P)-S_{\mathrm{eq}}(i,j)$ can take at most three possible values: $+\ln m$, $-\ln m$, or $0$.
\subsection*{Arboreal coefficients}
When there is a single energetic edge, for each rooted spanning tree \(T\in\Theta_r(G)\), the path action associated with the unique directed path from \(i\) to the root has only three possibilities, so we partition \(\Theta_r(G)\) according to the position of the energetic edge along \(T_i\) as follows:
\begin{equation}\label{e-partition}
\begin{aligned}
\Theta_+(i)&=\{T\in \Theta_r(G):\ z_2 \to z_1 \in E_{T_i}\},\\
\Theta_-(i)&=\{T\in \Theta_r(G):\ z_1 \to z_2 \in E_{T_i}\},\\
\Theta_0(i)&=\{T\in \Theta_r(G):\ z_1\to z_2 \notin E_{T_i},\ z_2\to z_1 \notin E_{T_i}\}.
\end{aligned}
\end{equation}
Therefore, \(\mathrm{e}^{-S(T_i)}\) depends only on which class \(T\) belongs to, that is,
\(\mathrm{e}^{-S(T_i)}=\mathrm{e}^{-S_{\mathrm{eq}}(i,r)}\) on \(\Theta_0(i)\),
\(\mathrm{e}^{-S(T_i)}=m\,\mathrm{e}^{-S_{\mathrm{eq}}(i,r)}\) on \(\Theta_+(i)\), and
\(\mathrm{e}^{-S(T_i)}=\frac{1}{m}\mathrm{e}^{-S_{\mathrm{eq}}(i,r)}\) on \(\Theta_-(i)\).
Grouping Eq.~\eqref{e-rhoness} accordingly yields
\begin{equation}\label{e-arho}
\rho_i(G)
=\mathrm{e}^{-S_{\mathrm{eq}}(i,r)}
\left(A_0(i)+mA_+(i)+\frac{1}{m}A_-(i)\right),
\end{equation}
where the \emph{arboreal coefficients} are
\begin{equation}\label{e-acoeffs}
A_0(i)\equiv\sum_{T\in\Theta_0(i)}\pr_{\Theta_r(G)}(T),\
A_+(i)\equiv\sum_{T\in\Theta_+(i)}\pr_{\Theta_r(G)}(T),\
A_-(i)\equiv\sum_{T\in\Theta_-(i)}\pr_{\Theta_r(G)}(T).
\end{equation}
Since \(\Theta_0(i)\sqcup\Theta_+(i)\sqcup\Theta_-(i)=\Theta_r(G)\), we have
\[
A_0(i)+A_+(i)+A_-(i)=1,\qquad \forall\, i\in V_G.
\]
Thus, the arboreal distribution is a normalized probability distribution over spanning trees rooted at the reference vertex. Importantly, it is not merely a mathematical construct, rather, it encodes physically relevant information about the response of systems driven far from equilibrium and also reveals interesting connections to loop-erased random walks \cite{cetiner2024jarzynski}.
\subsection*{The response under single energetic edge}
\begin{theorem}
Consider a Markov process initially at thermodynamic equilibrium and perturb it by modifying a single transition rate, e.g., $\ell(z_1\to z_2)\mapsto m\,\ell_{\mathrm{eq}}(z_1\to z_2)$. Then the normalized response ratio $K_{ij}(m)$ is a monotone function of $m$ for any $i$, $j$.
\end{theorem}
\begin{proof}
Using the path-centric reformulation, Eq.~\eqref{e-arho}, we start with the following ratio,
\begin{equation}
  \frac{\rho_i(G)}{\rho_j(G)}=\frac{\mathrm{e}^{-S_\text{eq}(i,r)}}{\mathrm{e}^{-S_\text{eq}(j,r)}}\times \frac{ A_0(i)+mA_+(i)+ A_-(i)/m}{A_0(j)+mA_+(j)+ A_-(j)/m} \,.
\end{equation}
The left hand side is nothing but $\pi_i / \pi_j= \rho_i(G)/\rho_j(G)$ and the right hand side has $\mathrm{e}^{-S_\text{eq}(i,r)} / \mathrm{e}^{-S_\text{eq}(j,r)}=\pi^\text{eq}_i/ \pi^\text{eq}_j$ and rearranging these terms gives the response as follows,
\begin{equation} \label{e-K}
K_{ij}(m)=\frac{A_0(i)+mA_+(i)+ A_-(i)/m }{A_0(j)+mA_+(j)+ A_-(j)/m }.
\end{equation}

We can now make a useful choice that puts the response function in its simplest form by setting $r=z_1$. Then the arboreal probabilities are defined on $\Theta_{z_1}$. Since every tree in $\Theta_{z_1}$ is rooted at $z_1$, the edge $z_1\to z_2$ cannot appear in any such tree, because the root has no outgoing edge. Hence all arboreal coefficients in Eq.~\eqref{e-K} are \emph{$m$-free}. In addition, Eq.~\eqref{e-partition} implies that
\[
A_-(i)=0 \qquad \text{for all } i,
\]
while normalization gives
\[
A_0(i)=1-A_+(i).
\]
Substituting these relations into Eq.~\eqref{e-K} yields the simplest form of the response function.
\begin{equation}\label{e-Ksim}
K_{ij}(m)=\frac{1+(m-1)A_+(i)}{1+(m-1)A_+(j)}.
\end{equation}
Differentiating Eq.~\eqref{e-Ksim} with respect to $m$ gives,
\begin{equation}\label{e-diffK}
\frac{dK_{ij}(m)}{dm}=\frac{A_+(i)-A_+(j)}{\bigl[1+(m-1)A_+(j)\bigr]^2}.
\end{equation}

The denominator is always strictly positive, while the numerator is independent of $m$. Therefore, the sign of $dK_{ij}/dm$ is independent of $m$, so $K_{ij}(m)$ must be monotone. More precisely, $K_{ij}(m)$ is increasing if $A_+(i)>A_+(j)$, decreasing if $A_+(i)<A_+(j)$.
\end{proof}

Thus, for any Markov process driven out of equilibrium by changing a single transition rate, the response can only be monotonic. As shown above, the arboreal coefficients distill away the nonequilibrium complexity that emerges once the system is driven away from equilibrium, and neatly capture the aftermath of the energetic perturbation.
\begin{figure}[t]
  \centering
  \includegraphics[width=\linewidth]{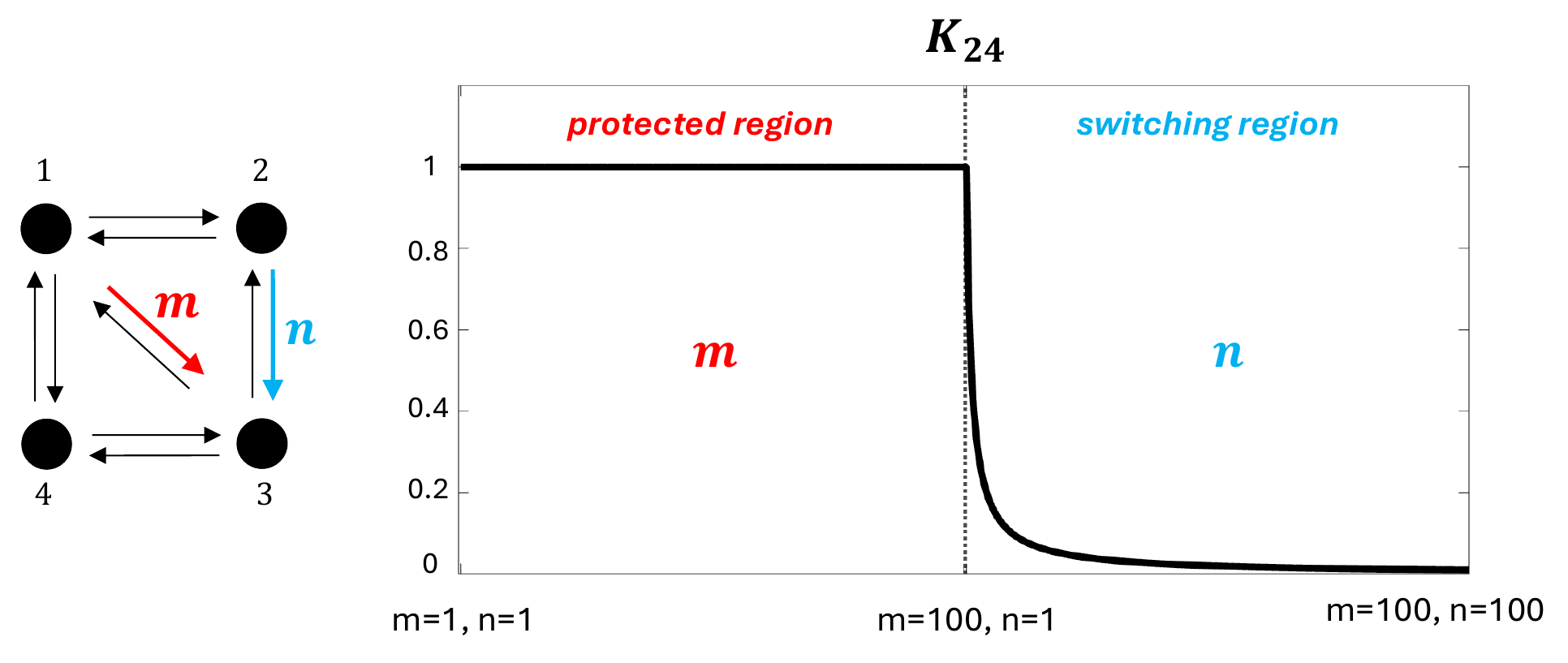}
  \caption{\textbf{A thermodynamic switch in a four-state Markov process.}
The network of Fig.~\ref{fig:Fig1}A is driven by two energetic edges, $1\to3$
(factor~$m$, red) and $2\to3$ (factor~$n$, cyan), starting from an equilibrium configuration that satisfies the protection condition $w(T_{\#6})=w(T_{\#7})$ (rates in Table~\ref{t-2}). The black curve traces the response $K_{24}$ as a two-stage protocol advances along the horizontal axis. In the \emph{protected region} (first stage: $n=1$ held fixed, $m$ ramped from $1$ to $100$) protection pins $\pi_2/\pi_4$ to its equilibrium value, so $K_{24}=1$ throughout. In the \emph{switching region} (second stage: $m=100$ held fixed, $n$ ramped through the edge $2\to3$) the response drops substantially. Energetic protection of the first edge together with a second energetic drive thus produces a switch far
from equilibrium.
}
\label{fig:Fig2}
\end{figure}
\subsection*{Protected Responses to Energetic Perturbations}
Eq.~\eqref{e-diffK} reveals another important consequence. If $A_+(i)=A_+(j)$, then the response becomes completely insensitive to the energetic perturbation: even as the driving strength $m$ is increased arbitrarily, the ratio $\pi_i/\pi_j$ remains fixed at its equilibrium value. In this sense, the response is \emph{protected} against energetic driving. As a concrete example, Fig.~\ref{fig:Fig1} shows a four-state Markov process in which the energetic edge $1\to 3$ is highlighted in red. In accordance with the analysis above, we choose the reference vertex as the source vertex of the energetic edge, $r=1$, so that Eqs. ~\eqref{e-Ksim} and \eqref{e-diffK} apply directly. We display all spanning trees rooted at vertex $1$ and label them by $T_{\#i}$ for convenience. We then study the response defined by the normalized steady-state ratio of occupation probabilities at vertices $2$ and $4$, $K_{24}(m)$. In accordance with Eq.~\eqref{e-partition}, the set of trees that contribute to the $A_+(i)$'s are given below for $i=2,4$,
\begin{equation}
    \Theta_+(2)=\{T_{\#7}, T_{\#8} \},\ \Theta_+(4)=\{T_{\#6}, T_{\#8} \}.
\end{equation}
Therefore, to achieve energetic protection of the response function against the energy expenditure at $1 \to 3$, we need to have $A_+(2)=A_+(4)$, equivalently, we need to have 
\begin{equation} \label{e-treecondition2}
    w(T_{\#6})=w(T_{\#7}).
\end{equation}

When this constraint holds as shown in Fig.~\ref{fig:Fig1}D, $K_{24}(m)$ is protected. Had we used the Matrix-Tree Theorem directly, we would have obtained the response as a rational function of the edge labels involving tedious monomial expressions (Eq.~\eqref{e-mtt}). By contrast, the present reformulation compresses the information contained in those exact but cumbersome monomials so effectively that the overall behavior of the response reduces to a simple binary comparison between the two trees highlighted in Fig.~\ref{fig:Fig1}. Yet all the analysis remains fully exact. \emph{Energetic protection} is the first key step toward constructing a thermodynamic switch.
\subsection*{Two-energetic edges}
Our formalism can be easily extended to handle two or more energetic edges. Since two energetic edges are already sufficient to design our switch, we consider that case next.
\begin{theorem}
Consider a Markov process maintained in a nonequilibrium steady state by continuous energy expenditure at two energetic edges,
$\ell(z_1 \to z_2)=m\,\ell_{\mathrm{eq}}(z_1 \to z_2),\ \ell(z_3 \to z_4)=n\,\ell_{\mathrm{eq}}(z_3 \to z_4)$. If the system is perturbed further by varying only one of these energetic edges while keeping the other fixed, then the response function \(K_{ij}(m,n)\) can only be monotonic.
\end{theorem}
\begin{proof}
In the presence of two energetic edges, Eq.~\eqref{e-path_entropies2} extends to
\begin{equation}\label{e-path_entropies_two}
S(P)=S_{\mathrm{eq}}(i,j)+\sigma_{12}(P)\ln(m)+\sigma_{34}(P)\ln(n).
\end{equation}
Or, in a form more compatible with our formalism,
\begin{equation}\label{e-path_entropies_three}
\mathrm{e}^{-S(P)}=\mathrm{e}^{-S_{\mathrm{eq}}(i,j)}\,m^{-\sigma_{12}(P)}n^{-\sigma_{34}(P)}.
\end{equation}
There are nine classes of arboreal trees, corresponding to
\[
(\alpha,\beta)\in (+,0,-)^2,
\]
where \(\alpha=+\) if \(z_2\to z_1\in E_{T_i}\), \(\alpha=-\) if \(z_1\to z_2\in E_{T_i}\), and \(\alpha=0\) otherwise; the same definition applies to \(\beta\) for the second energetic edge. We define
\begin{equation}
A_{\alpha,\beta}(i)
=
\sum_{T\in\Theta_r(G),\; T_i\text{ in class }(\alpha,\beta)}
\mathrm{Pr}_{\Theta_r(G)}(T).
\end{equation}
Then
\begin{align}
\rho_i
&=\mathrm{e}^{-S_{\mathrm{eq}}(i,r)}
\Bigl(
A_{0,0}(i)
+ m\,A_{+,0}(i)
+ \frac{A_{-,0}(i)}{m}
+ n\,A_{0,+}(i)
+ \frac{A_{0,-}(i)}{n}
\nonumber\\
&\qquad\quad
+ \frac{m}{n}\,A_{+,-}(i)
+ \frac{n}{m}\,A_{-,+}(i)
+ \frac{A_{-,-}(i)}{mn}
+ mn\,A_{+,+}(i)
\Bigr).
\end{align}
If we now choose \(z_1\) as the reference vertex, a further simplification occurs as before. Since no tree rooted at \(z_1\) can contain the edge \(z_1\to z_2\), the arboreal distribution becomes \(m\)-free (but not \emph{n-free}), and
\[
A_{-,0}(i)=A_{-,+}(i)=A_{-,-}(i)=0.
\]
Hence
\begin{equation}\label{e-rho-two-edge-simple}
\rho_i
=
\mathrm{e}^{-S_{\mathrm{eq}}(i,r)}
\bigl(\mathcal{B}_0(i,n)+m\,\mathcal{B}_+(i,n)\bigr),
\end{equation}
where
\begin{equation}\label{e-B-def}
\begin{aligned}
\mathcal{B}_0(i,n)
&=
A_{0,0}(i)+n\,A_{0,+}(i)+\frac{A_{0,-}(i)}{n},\\[3pt]
\mathcal{B}_+(i,n)
&=
A_{+,0}(i)+n\,A_{+,+}(i)+\frac{A_{+,-}(i)}{n}.
\end{aligned}
\end{equation}

The information-theoretic response is therefore
\begin{equation}\label{e-Kmn}
K_{ij}(m,n)
=
\frac{\pi_i/\pi_j}{\pi_i^{\mathrm{eq}}/\pi_j^{\mathrm{eq}}}
=
\frac{\mathcal{B}_0(i,n)+m\,\mathcal{B}_+(i,n)}
{\mathcal{B}_0(j,n)+m\,\mathcal{B}_+(j,n)} .
\end{equation}
Setting \(n=1\) recovers the single-edge case. Indeed,
\[
\mathcal{B}_0(i,1)=A_0(i),
\qquad
\mathcal{B}_+(i,1)=A_+(i),
\]
so Eq.~\eqref{e-Kmn} reduces to the one-edge expression, Eq.~\eqref{e-Ksim}. The energetic protection is then equivalent to $K_{ij}(m,1)=1$ for all $m$.

If \(n\) is held fixed and \(m\) is varied, then
\begin{equation}\label{e-partialKmn}
\frac{\partial K_{ij}(m,n)}{\partial m}
=
\frac{
\mathcal{B}_+(i,n)\mathcal{B}_0(j,n)
-
\mathcal{B}_+(j,n)\mathcal{B}_0(i,n)
}
{
\bigl(\mathcal{B}_0(j,n)+m\,\mathcal{B}_+(j,n)\bigr)^2
}.
\end{equation}
The denominator is again strictly positive, while the numerator in Eq.~\eqref{e-partialKmn} is independent of \(m\). Hence the partial derivative has a fixed sign, and the response is monotonic in \(m\). Thus, even though the system must continuously dissipate energy at two edges to maintain a nonequilibrium steady state, varying one energetic edge while keeping the other fixed can produce only a monotonic response.
\end{proof}

This is precisely the mechanism needed to construct a switch-like response function. Suppose the two energetic edges, parameterized by \((m,n)\), can be externally controlled and varied in time according to a prescribed protocol. We first choose parameters such that the response is protected against perturbations in \(m\), as in the first part. Starting from equilibrium and increasing \(m\), the response therefore remains pinned to its equilibrium value,
$\pi_i/\pi_j=\pi_i^\mathrm{eq}/ \pi_j^{\mathrm{eq}}$, even while the system is driven arbitrarily far from thermodynamic equilibrium. At a chosen moment, we then hold \(m\) fixed and activate the second energetic edge, controlled by \(n\). By the second theorem, the response can vary with \(n\) only monotonically when \(m\) is fixed. For a suitable choice of parameters, this monotonic variation can be made switch-like, as shown in Fig.~\ref{fig:Fig2}. Energetic protection, combined with conditional monotonicity far from equilibrium, therefore provides a simple route to a thermodynamic switch.
\begin{figure}[t]
  \centering
  \includegraphics[width=\linewidth]{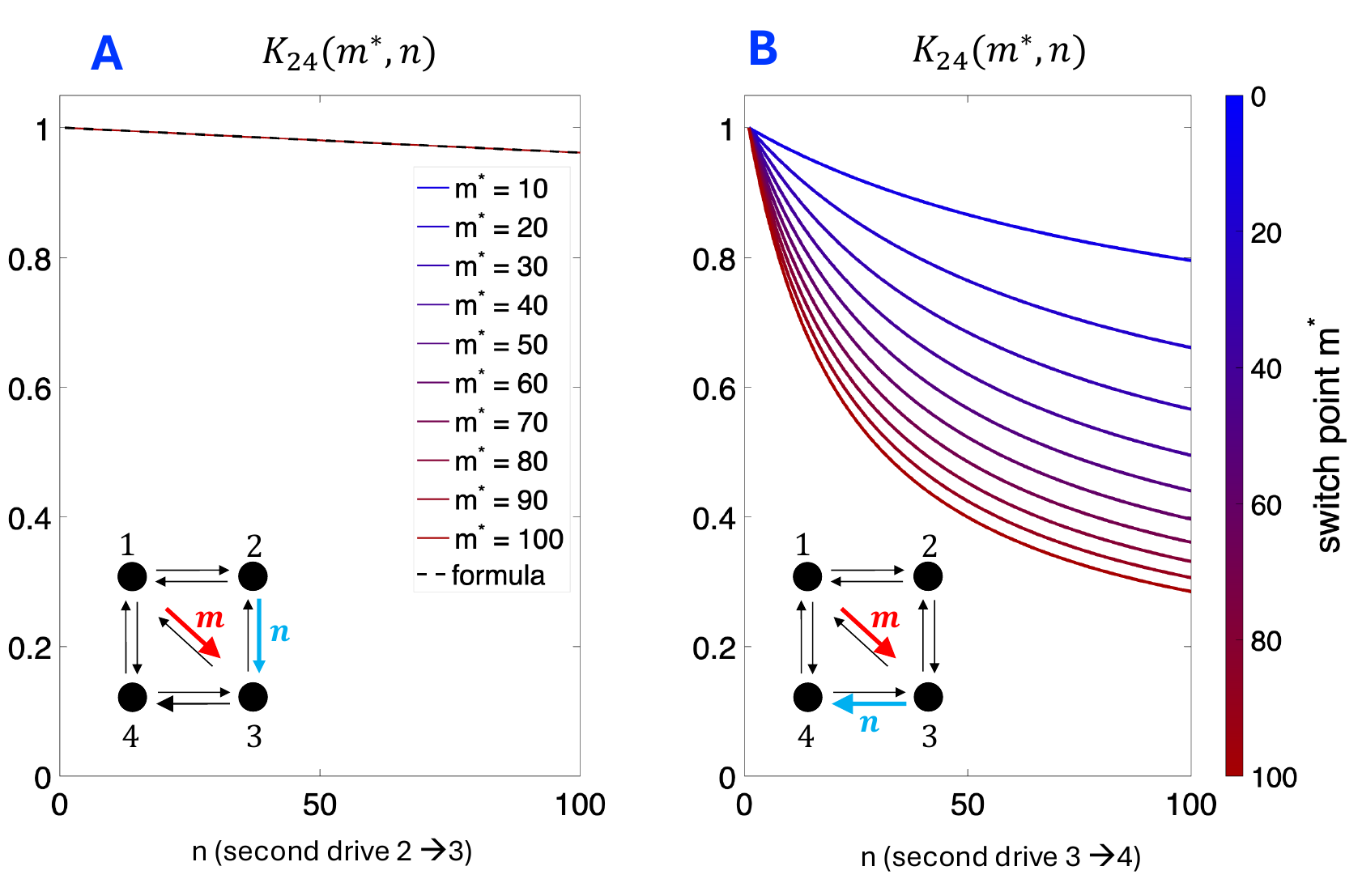}
  \caption{\textbf{Whether the second-stage switch is protocol-independent
depends on the choice of second energetic edge.}
In both panels the first edge $1\to3$ is ramped to a stopping (switch)
value $m^{*}$ and then held fixed while the second edge is ramped
through~$n$; curves are coloured by $m^{*}$, and the insets show the two
energetic edges ($m$ red, $n$ cyan). Both panels use the same equilibrium rates,
which satisfy the protection condition $w(T_{\#6})=w(T_{\#7})$ with regard to the first energetic edge ($1\to3$), and differ only in which edge carries the second drive (Table~\ref{t-3}). (\textbf{A})~Second edge $2\to3$. Energetic protection extends from the single slice $n=1$ to the full identity $W(n)=0$ for all~$n$ (Supplement), so the second-stage response is independent of the switch point: every curve collapses onto the closed form
$K_{24}(n)=1/[\,1+(n-1)\theta_{\mathrm{eq}}\,]$, with
$\theta_{\mathrm{eq}}=\ell_{\mathrm{eq}}(2\to3)/[\,\ell_{\mathrm{eq}}(2\to3)
+\ell_{\mathrm{eq}}(2\to1)\,]$ (dashed).
(\textbf{B})~Second edge $3\to4$. Protection now yields only $W(1)=0$ and not
$W(n)=0$, so the curves no longer collapse: the response depends on the
switch point~$m^{*}$ (colour bar), and $m^{*}$ can be chosen to deepen
the switch.}
  \label{fig:Fig3}
\end{figure}
\subsection*{Optimal m to turn on the second drive}
In a typical application of energetic protection, where the first energetic edge produces a switch-like response as in Fig.~\ref{fig:Fig2}, a natural question arises whether the point at which this first drive is stopped affects the subsequent response. More concretely, suppose that we increase the first energetic parameter $m$ up to some value $m^*$, hold it fixed, and then begin increasing the second energetic parameter $n$. Is there an optimal choice of stopping $m^*$ that improves the dynamical range of the overall response? The answer is yes, and it depends on the numerator of Eq.~\eqref{e-partialKmn}. Define the numerator as
\begin{equation}
W(n)\equiv \mathcal{B}_+(i,n)\mathcal{B}_0(j,n)
-\mathcal{B}_+(j,n)\mathcal{B}_0(i,n).
\end{equation}
Energetic protection is a condition imposed on the slice \(n=1\), and it guarantees only that \(W(1)=0\). Indeed,
\begin{equation}
  \begin{split}
  W(1)
  &= \mathcal{B}_+(i,1)\mathcal{B}_0(j,1)
     -\mathcal{B}_+(j,1)\mathcal{B}_0(i,1) \\
  &= A_+(i)A_0(j)-A_+(j)A_0(i) \\
  &= A_+(i)\bigl(1-A_+(j)\bigr)
     -A_+(j)\bigl(1-A_+(i)\bigr) \\
  &= A_+(i)-A_+(j).
  \end{split}
\end{equation}
Thus, the energetic-protection condition $(A_+(i)=A_+(j))$ gives exactly
$W(1)=0.$ This is only a cancellation at the slice $(n=1)$. It does not, by itself, promote  $W(n)$ to the identically vanishing function
$ W(n)= 0$ for all $n$. Whether this stronger identity holds is a separate question and it depends on the particular choice of energetic edges, and therefore has a topological flavor. If it holds, then the subsequent response is independent of the value $m^*$ at which the first drive is stopped and the second drive is started. That is, there is no distinguished value of $m^*$, and all stopping values of $m^*$ give exactly the same second-stage curve $K(m^*,n)$. For example, for the two energetic edges ($(1\to 3)$ and $(2\to 3)$) used in Fig.~\ref{fig:Fig2}, one can show that energetic protection extends from the single slice $n=1$ to the full identity $W(n)=0$ for all $n$ (See Supplementary Information). Consequently, $K_{24}(m,n)$ behaves in the same way regardless of the value at which $m$ is held fixed before increasing $n$. Indeed, one can show that in this case $K_{24}(m,n)=K_{24}(n)=1/(1+(n-1)\theta_{\text{eq}})$, where $\theta_{\text{eq}}=\ell_{\mathrm{eq}}(2\to 3)/(\ell_{\mathrm{eq}}(2\to 3)+\ell_{\mathrm{eq}}(2\to 1))$. In Fig.~\ref{fig:Fig3}A, we show the second stage of the response $K_{24}$ for various stopping values of $m^*$ under the same protocol used in Fig.~\ref{fig:Fig2}. All response curves are identical, and they collapse onto $K_{24}(n)=1/(1+(n-1)\theta_{\text{eq}})$.

This protocol-independence is not generic, however. If we choose a different second energetic edge, for example $(3\to 4)$, then energetic protection still gives $W(1)=0$, but it does not extend to $W(n)=0$ for all $n$. In that case, the response can depend on the stopping value $m^*$. Since $K_{24}$ is monotone in $m$, one can ramp $m$ as far as the driving budget allows to deepen the switch. Fig.~\ref{fig:Fig3}B shows this idea. Note that the difference between panel A and B is only the choice of the second energetic edge and is independent of the rate constants, as long as they respect the energetic protection constraints. The parameter set used in the creation of Fig.~\ref{fig:Fig3} is provided in Table~\ref{t-3} in the Supplement.

\subsection*{Nonmonotonic response}
\begin{figure}[t]
  \centering
  \includegraphics[width=\linewidth]{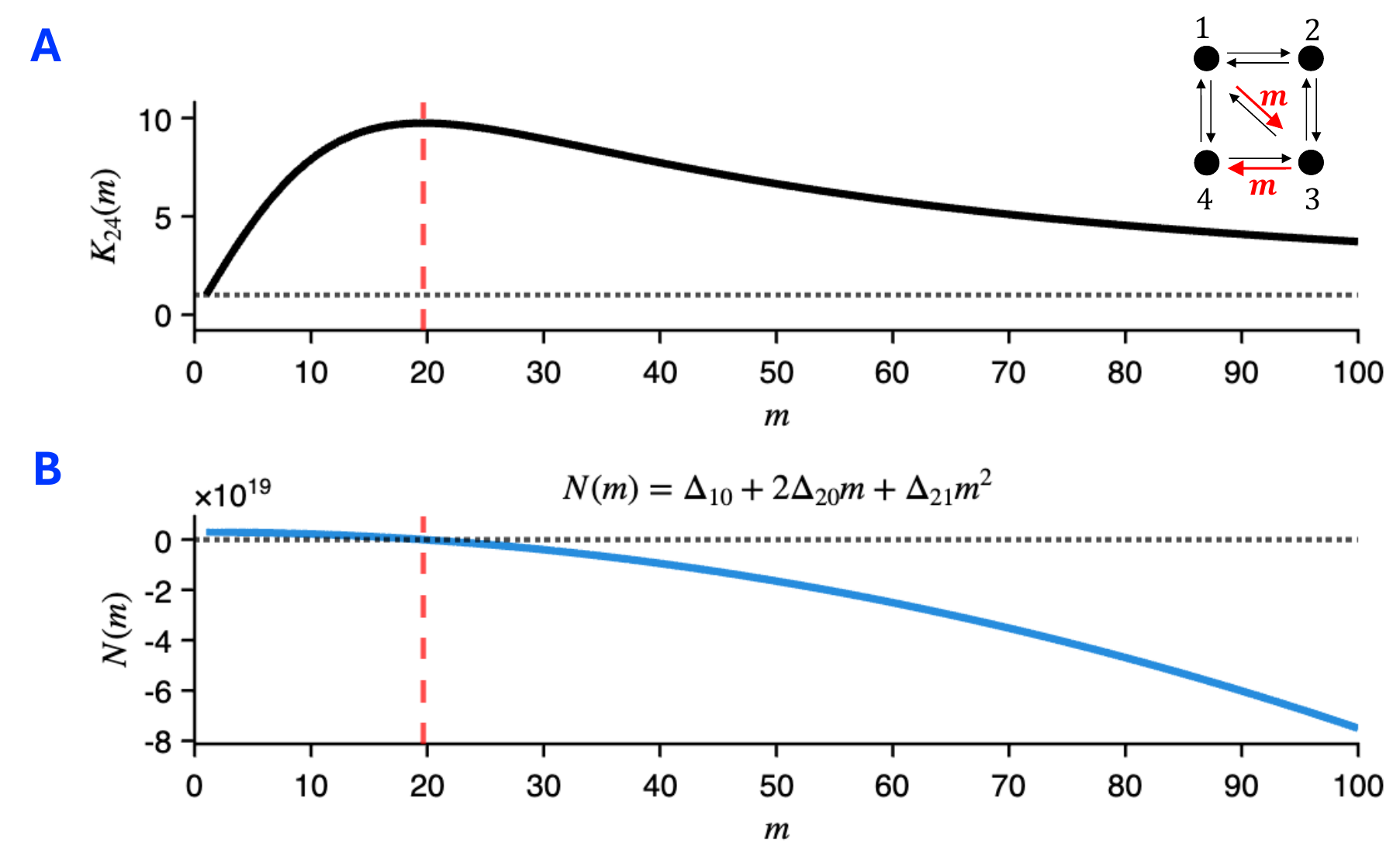}
  \caption{\textbf{Two coupled energetic edges produce a single
nonmonotonic bump.} The energetic edges $1\to3$ and $3\to4$ are perturbed together by a common factor~$m$ from an equilibrium configuration ($m=1$; rates in Table~\ref{t-4}). (\textbf{A})~The normalized response $K_{24}(m)$ rises and then falls, with a single interior maximum. (\textbf{B})~Its slope obeys $\operatorname{sgn}K_{24}'(m)=\operatorname{sgn}N(m)$ with $N(m)=\Delta_{10}+2\Delta_{20}m+\Delta_{21}m^{2}$ [Eq.~\eqref{e-Ndef}]. Here $N(1)>0$ while $N(m)<0$ as $m\to\infty$, so $N$ crosses zero exactly once (dashed line); this root locates the maximum in (\textbf{A}). Because the coefficients $a_k,b_k$ are nonnegative, $N(m)$ can change sign at most once, so two coupled energetic edges never yield more than a single bump.}
\label{fig:Fig4}
\end{figure}
Changing the rate of a single energetic edge, while holding the others fixed, cannot make the response nonmonotonic. Nonmonotonicity can, however, arise when two energetic edges are coupled and varied together through a relation \(n=n(m)\). For concreteness we take \(n=m\), so that both edges are driven by the same parameter. All \(m\)-dependent contributions must then be retained---both the spanning-tree weights and the path contributions entering the steady-state probabilities (see Supplement). The response takes the form
\begin{equation}\label{e-response_quadraticnew}
K_{ij}(m)=\frac{a_0+a_1m+a_2m^2}{b_0+b_1m+b_2m^2},
\end{equation}
where the coefficients \(a_k\) and \(b_k\) are nonnegative sums of equilibrium
spanning-tree weights (Supplement). Differentiating Eq.~\eqref{e-response_quadraticnew} yields
\begin{equation}\label{e-response_derivative_m}
\frac{dK_{ij}}{dm}
=\frac{(a_1b_0-a_0b_1)+2(a_2b_0-a_0b_2)m+(a_2b_1-a_1b_2)m^2}
{(b_0+b_1m+b_2m^2)^2}.
\end{equation}
Because the denominator is a square, \(\operatorname{sgn}K_{ij}'(m)=\operatorname{sgn}N(m)\),
where
\begin{equation}\label{e-Ndef}
N(m)\equiv\Delta_{10}+2\Delta_{20}m+\Delta_{21}m^2,
\qquad
\Delta_{kl}\equiv
\begin{vmatrix} a_k & a_l \\ b_k & b_l \end{vmatrix}
=a_kb_l-a_lb_k.
\end{equation}
Thus \(N(m)\) alone sets the monotonicity: \(K_{ij}\) rises where \(N(m)>0\), falls where \(N(m)<0\), and is stationary where \(N(m)=0\). The response is nonmonotonic precisely when \(N(m)\) changes sign on \((1,\infty)\).

Two boundary slopes frame the picture: the onset, \(\operatorname{sgn}K_{ij}'(1)=\operatorname{sgn}N(1)\) with \(N(1)=\Delta_{10}+2\Delta_{20}+\Delta_{21}\), and the far field, \(\operatorname{sgn}N(m)\to\operatorname{sgn}\Delta_{21}\) as \(m\to\infty\). When these disagree, i.e.\ \(N(1)\,\Delta_{21}<0\), the quadratic \(N(m)\) begins with one sign and ends with the other, so it has exactly one root in \((1,\infty)\). The response then has a single interior extremum---a maximum if \(N\) goes \(+\!\to\!-\), a minimum if \(-\!\to\!+\). Fig.~\ref{fig:Fig4}A shows such a response, \(K_{24}(m)\), obtained by
perturbing the two energetic edges \(1\to 3\) and \(3\to 4\) together from an equilibrium configuration. Its monotonicity is governed by \(N(m)\), plotted in
Fig.~\ref{fig:Fig4}B: here \(N(1)>0\) while \(N(m)<0\) as \(m\to\infty\), so \(N\) crosses zero once, and this single root produces the single bump---the maximum of \(K_{24}(m)\).

One might expect a second route to nonmonotonicity when the boundary slopes instead \emph{agree}, \(\operatorname{sgn}N(1)=\operatorname{sgn}\Delta_{21}\): a double bump would then require two roots of \(N\) inside \((1,\infty)\), i.e.\ a vertex \(m_v=-\Delta_{20}/\Delta_{21}>1\) together with a positive discriminant
\(\Delta_{20}^2-\Delta_{21}\Delta_{10}>0\). With nonnegative coefficients, however, this never occurs. Two roots, say, $r_2>r_1>1$ would give us   $r_2r_1=\Delta_{10}/\Delta_{21}>0$ and $r_1+r_2=-2\Delta_{20}/\Delta_{21}>0$ and would force the sign pattern \((\Delta_{10},\Delta_{21},\Delta_{20})=(+,+,-)\) or \((-,-,+)\), which is impossible as it leads to contradiction (see Supplement). Hence \(N\) changes sign at most once on \((1,\infty)\) and two coupled energetic edges can produce at most a single bump. A double wiggle requires at least three energetic edges. How the topology of energy expenditure shapes nonmonotonicity with more energetic edges and more general couplings \(n=n(m)\) is left for future work. The parameters used in Fig.~\ref{fig:Fig4} are listed in the Supplement.
\section*{Conclusions}
Thermodynamic equilibrium is the natural state toward which many systems relax when left alone. At the same time, thermodynamic equilibrium imposes nontrivial limits and tradeoffs on information-processing capabilities, some of which can be relieved only by energy expenditure. Operating away from equilibrium is therefore a more special situation, and sustaining it requires a continuous cost. In real biochemical systems, this cost is typically supplied by free-energy-consuming processes such as ATP hydrolysis, and importantly, it is not distributed randomly across all reactions but is instead coupled to particular transitions. In our graph-theoretic approach to Markov processes, we capture this through energetic edges.

Our results show that where and how energy is expended matters. More than that, even when the underlying system is driven arbitrarily far from thermodynamic equilibrium, new limitations survive. For example, if energy is expended at only a single transition, then no matter how far the system is driven from equilibrium, a nonmonotonic response cannot arise. Remarkably, the full combinatorial complexity of the nonequilibrium steady state collapses to the comparison of just two quantities, the arboreal coefficients \(A_+(i)\) and \(A_+(j)\), whose ordering alone determines whether the response increases or decreases monotonically. In this sense, nonequilibrium function is controlled not only by free-energy consumption, but by how that free-energy consumption is embedded in the underlying graph, which is exactly captured by the arboreal coefficients.

A particularly interesting consequence is the emergence of a new nonequilibrium phenomenon, \emph{energetic protection}. In addition to time-reversal symmetry, which at thermodynamic equilibrium imposes algebraic constraints on cycle rates, energetic protection requires new algebraic constraints involving equalities between the weights of certain spanning trees of the underlying graph. When these constraints hold, a chosen steady-state ratio remains fixed at its equilibrium value even though the system is driven arbitrarily far from equilibrium. This shows that information-theoretic invariances can survive deep in the nonequilibrium regime. In this sense, energetic protection provides a form of nonequilibrium robustness: dissipation is present, but its functional effect on a chosen observable is exactly neutralized by algebraic constraints arising from relations among certain spanning trees.

With two energetic edges, energetic protection also enables a thermodynamic switch: the response can be held at its equilibrium value for as long as desired and then changed substantially by activating a second drive. This shows that energetic driving can do more than simply amplify or suppress function; it can also store, release, and control function in a programmable way. More broadly, our graph-theoretic formulation gives path-centric access to the fabric of nonequilibrium steady states and suggests a general route for understanding how the localization of free-energy transduction shapes the logic of information processing in stochastic systems.

\bibliographystyle{unsrt} 
\bibliography{refs}
\newpage

\setcounter{equation}{0}\renewcommand{\theequation}{S\arabic{equation}}
\setcounter{table}{0}\renewcommand{\thetable}{S\arabic{table}}
\begingroup
\centering
{\Large\bfseries Supplementary Information}\\[4pt]
{\large Energetic Protection, Monotonicity and Switching Far from Equilibrium}\\[3pt]
\endgroup
\vspace{1.5em}
\section*{Contents of the Supplement}
\noindent
This Supplement collects the derivations and parameter values referenced in the
main text. It is organized into three self-contained sections.

\medskip
\noindent\textbf{\hyperref[si:S1]{S1.\enspace The quadratic response.}}\\
With both energetic edges driven by a common factor ($n=m$), the normalized
response $K_{ij}(m)$ is shown to be a ratio of two quadratics in $m$ with
non-negative coefficients, establishing Eq.~\eqref{e-response_quadraticnew} of
the main text.

\smallskip
\noindent\textbf{\hyperref[si:S2]{S2.\enspace Impossibility of a double bump.}}\\
The numerator $N(m)$ of $\partial_m K_{ij}$ is shown to admit at most one
positive root, so two coupled energetic edges can turn the response over at most
once.

\smallskip
\noindent\textbf{\hyperref[si:S3]{S3.\enspace Energetic protection and the
protocol-independent switch.}}\\
For driving along $1\!\to\!3$ with second energetic edge $2\!\to\!3$, the switch
kernel is shown to vanish identically, $W(n)\equiv0$, so the second-stage
response collapses onto the closed form
$K_{24}(n)=[\,1+(n-1)\,\theta_{\mathrm{eq}}\,]^{-1}$, independent of the
$1\!\to\!3$ protocol.

\medskip
\noindent
Parameter values for all figures of the main text are collected in
Tables~\ref{t-1}--\ref{t-4}.

\bigskip
\hrule
\bigskip
\phantomsection\label{si:S1}
\section*{S1.\quad Derivation of the quadratic response, Eq.~\eqref{e-response_quadraticnew}}

When the two energetic edges are coupled, so that
$\ell(z_1\!\to\!z_2)=m\,\ell_{\mathrm{eq}}(z_1\!\to\!z_2)$ and
$\ell(z_3\!\to\!z_4)=m\,\ell_{\mathrm{eq}}(z_3\!\to\!z_4)$, every $m$-dependent
contribution must be collected---both from the spanning trees and from the simple
paths that enter the steady-state probabilities. For a spanning tree
$T\in\Theta_r(G)$, let
\begin{equation}\label{e-uT}
u(T)=\mathbf{1}_{\{z_1\to z_2\,\in\,E_T\}}
    +\mathbf{1}_{\{z_3\to z_4\,\in\,E_T\}}\in\{0,1,2\}
\end{equation}
count how many energetic edges $T$ contains, where $\mathbf{1}_{\{\cdot\}}$ is
the indicator function. Since each energetic edge present multiplies its
equilibrium rate by $m$, the tree weight factorizes as
\begin{equation}\label{e-treeweight}
w(T)=m^{u(T)}\,w_{\mathrm{eq}}(T),
\qquad
w_{\mathrm{eq}}(T)=\prod_{e\,\in\,E_T}\ell_{\mathrm{eq}}(e).
\end{equation}
To track the full $m$-dependence of trees and paths together, define
\begin{equation}\label{e-xT}
x(T,i)=u(T)-\sigma_{12}(T_i)-\sigma_{34}(T_i).
\end{equation}
Because $T_i\subseteq T$, an energetic edge lying on the path is necessarily
present in the tree, so each of the differences
$\mathbf{1}_{\{z_1\to z_2\in E_T\}}-\sigma_{12}(T_i)$ and
$\mathbf{1}_{\{z_3\to z_4\in E_T\}}-\sigma_{34}(T_i)$ takes only the values $0$
or $1$. Hence $x(T,i)\in\{0,1,2\}$. Combining
Eq.~\eqref{e-path_entropies_three} with Eq.~\eqref{e-treeweight} gives
\begin{equation}\label{e-nonmono-ratio}
\frac{\rho_i}{\rho_j}
=\frac{\mathrm{e}^{-S_{\mathrm{eq}}(i,r)}}{\mathrm{e}^{-S_{\mathrm{eq}}(j,r)}}\,
 \frac{\displaystyle\sum_{T\in\Theta_r(G)}m^{x(T,i)}\,w_{\mathrm{eq}}(T)}
      {\displaystyle\sum_{T\in\Theta_r(G)}m^{x(T,j)}\,w_{\mathrm{eq}}(T)}
=\frac{\mathrm{e}^{-S_{\mathrm{eq}}(i,r)}}{\mathrm{e}^{-S_{\mathrm{eq}}(j,r)}}\,
 \frac{\displaystyle\sum_{k\in X}a_k\,m^{k}}{\displaystyle\sum_{k\in X}b_k\,m^{k}},
\end{equation}
where $X=\{0,1,2\}$ and the second equality groups the trees by their common
value of $x$, with non-negative coefficients
\begin{equation}\label{e-akbk}
a_k=\!\!\sum_{\substack{T\in\Theta_r(G)\\ x(T,i)=k}}\!\! w_{\mathrm{eq}}(T),
\qquad
b_k=\!\!\sum_{\substack{T\in\Theta_r(G)\\ x(T,j)=k}}\!\! w_{\mathrm{eq}}(T).
\end{equation}
Since $\pi_i/\pi_j=\rho_i/\rho_j$ and
$\mathrm{e}^{-S_{\mathrm{eq}}(i,r)}/\mathrm{e}^{-S_{\mathrm{eq}}(j,r)}
=\pi_i^{\mathrm{eq}}/\pi_j^{\mathrm{eq}}$, the equilibrium prefactor cancels in
the normalized response, leaving
\begin{equation}\label{e-Kquad-si}
K_{ij}(m)=\frac{\sum_{k\in X}a_k\,m^{k}}{\sum_{k\in X}b_k\,m^{k}}
=\frac{a_0+a_1m+a_2m^2}{b_0+b_1m+b_2m^2},
\end{equation}
which is Eq.~\eqref{e-response_quadraticnew} of the main text.

\phantomsection\label{si:S2}
\section*{S2.\quad Impossibility of two positive roots for two coupled edges}

Recall from the main text that
$\operatorname{sgn}K_{ij}'(m)=\operatorname{sgn}N(m)$ with
\begin{equation}\label{e-Nsi}
N(m)=\Delta_{10}+2\Delta_{20}\,m+\Delta_{21}\,m^2,
\qquad
\Delta_{kl}=a_kb_l-a_lb_k .
\end{equation}
A double bump on $(1,\infty)$ would require $N(m)$ to have two positive roots.
Suppose, for contradiction, that such roots exist, $r_2>r_1>0$. By Vi\`ete's
relations,
\begin{equation}\label{e-vieta}
r_1r_2=\frac{\Delta_{10}}{\Delta_{21}}>0,
\qquad
r_1+r_2=-\frac{2\Delta_{20}}{\Delta_{21}}>0,
\end{equation}
so the ordered triple $(\Delta_{10},\Delta_{21},\Delta_{20})$ must carry the sign
pattern $(+,+,-)$ or $(-,-,+)$. We show that each pattern is contradictory.
Throughout we use that the $b_k$ are strictly positive, so dividing a minor
$\Delta_{kl}=a_kb_l-a_lb_k$ by $b_kb_l>0$ preserves its sign and yields
$a_k/b_k-a_l/b_l$.

\medskip
\noindent\emph{Pattern $(\Delta_{10},\Delta_{21},\Delta_{20})=(+,+,-)$.}
\begin{equation}\label{e-caseA}
\begin{aligned}
\Delta_{10}=a_1b_0-a_0b_1>0
  &\;\Longrightarrow\; \tfrac{a_1}{b_1}-\tfrac{a_0}{b_0}>0,\\
\Delta_{21}=a_2b_1-a_1b_2>0
  &\;\Longrightarrow\; \tfrac{a_2}{b_2}-\tfrac{a_1}{b_1}>0,\\
\Delta_{20}=a_2b_0-a_0b_2<0
  &\;\Longrightarrow\; \tfrac{a_2}{b_2}-\tfrac{a_0}{b_0}<0.
\end{aligned}
\end{equation}
Adding the first two inequalities gives $a_2/b_2-a_0/b_0>0$, contradicting the
third.

\medskip
\noindent\emph{Pattern $(\Delta_{10},\Delta_{21},\Delta_{20})=(-,-,+)$.}
The same manipulation reverses every inequality,
\begin{equation}\label{e-caseB}
\begin{aligned}
\Delta_{10}<0 &\;\Longrightarrow\; \tfrac{a_1}{b_1}-\tfrac{a_0}{b_0}<0,\\
\Delta_{21}<0 &\;\Longrightarrow\; \tfrac{a_2}{b_2}-\tfrac{a_1}{b_1}<0,\\
\Delta_{20}>0 &\;\Longrightarrow\; \tfrac{a_2}{b_2}-\tfrac{a_0}{b_0}>0,
\end{aligned}
\end{equation}
and adding the first two now gives $a_2/b_2-a_0/b_0<0$, again contradicting the
third. Hence $N(m)$ cannot possess two positive roots: it changes sign at most
once on $(1,\infty)$, and two coupled energetic edges can produce at most a
single bump.

\phantomsection\label{si:S3}
\section*{S3.\quad Derivation of $W(n)\equiv0$ under energetic protection}

Throughout this section it is convenient to relabel the equilibrium rates as
\begin{equation}\label{e-relabel}
\begin{gathered}
a=\ell_{\mathrm{eq}}(2\to1),\quad
b=\ell_{\mathrm{eq}}(2\to3),\quad
c=\ell_{\mathrm{eq}}(3\to1),\quad
d=\ell_{\mathrm{eq}}(3\to2),\\[2pt]
e=\ell_{\mathrm{eq}}(3\to4),\quad
f=\ell_{\mathrm{eq}}(4\to1),\quad
g=\ell_{\mathrm{eq}}(4\to3).
\end{gathered}
\end{equation}
For every spanning tree rooted at vertex~$1$, Table~\ref{t-1} lists the unique
simple paths from vertex~$2$ ($T_2$) and vertex~$4$ ($T_4$) to the root,
together with their $(\alpha,\beta)$ labels: $\alpha=+$ if $z_2\to z_1\in
E_{T_i}$, $\alpha=-$ if $z_1\to z_2\in E_{T_i}$, and $\alpha=0$ otherwise, with
the same convention defining $\beta$ for the second energetic edge. These labels
identify the trees that contribute to each $A_{\alpha,\beta}(i)$ for $i=2,4$.
Because the second energetic edge $2\to3$ can lie in a tree rooted at $1$, the
tree weights $w(T)$ and hence the coefficients $A_{\alpha,\beta}(i)$ retain
their dependence on $n$.

Combining Eq.~\eqref{e-B-def} with Table~\ref{t-1}, we obtain for vertex~$2$
\begin{equation}\label{e-B-vertex2}
\begin{aligned}
\mathcal{B}_0(2,n)
&=\frac{1}{w(\Theta_1(G))}\bigl(aef+adf+adg+acf+acg+bef\bigr),\\[3pt]
\mathcal{B}_+(2,n)
&=\frac{1}{w(\Theta_1(G))}\bigl(bcf+bcg\bigr).
\end{aligned}
\end{equation}
Similarly, for vertex~$4$,
\begin{equation}\label{e-B-vertex4}
\begin{aligned}
\mathcal{B}_0(4,n)
&=\frac{1}{w(\Theta_1(G))}\bigl(n\,(bef+bcf+adg)+aef+adf+acf\bigr),\\[3pt]
\mathcal{B}_+(4,n)
&=\frac{1}{w(\Theta_1(G))}\bigl(acg+n\,bcg\bigr).
\end{aligned}
\end{equation}
Therefore, $W(n)=\mathcal{B}_+(2,n)\,\mathcal{B}_0(4,n)-\mathcal{B}_+(4,n)\,\mathcal{B}_0(2,n)$ can be written as follows,
\begin{equation}\label{e-Wsplit}
W(n)=\frac{1}{w(\Theta_1(G))^{2}}\bigl(W_0+n\,W_1\bigr),
\end{equation}
with
\begin{align}
W_0&=\bigl(bcf+bcg\bigr)\bigl(aef+adf+acf\bigr)
   -acg\bigl(aef+adf+adg+acf+acg+bef\bigr),\label{e-W0}\\[3pt]
W_1&=\bigl(bcf+bcg\bigr)\bigl(bef+bcf+adg\bigr)
   -bcg\bigl(aef+adf+adg+acf+acg+bef\bigr).\label{e-W1}
\end{align}
We now show that the energetic-protection condition $w(T_{\#6})=w(T_{\#7})$,
i.e.\ $acg=bcf$, or equivalently $ag=bf$, forces both $W_0=0$ and $W_1=0$, so
that $W(n)\equiv0$.

\subsubsection*{Vanishing of $W_0$}
Factoring the first product of Eq.~\eqref{e-W0} and dividing by $c>0$,
\begin{equation}\label{e-W0-step}
\frac{W_0}{c}=abf\,(f+g)(c+d+e)
-ag\bigl(aef+adf+adg+acf+acg+bef\bigr).
\end{equation}
Expanding the two products,
\begin{equation}\label{e-W0-expand}
\begin{aligned}
abf\,(f+g)(c+d+e)
&=abcf^2+abdf^2+abef^2+abcfg+abdfg+abefg,\\[3pt]
ag\bigl(aef+adf+adg+acf+acg+bef\bigr)
&=a^2efg+a^2dfg+a^2dg^2+a^2cfg+a^2cg^2+abefg.
\end{aligned}
\end{equation}
Taking the difference and factoring, $W_0$ collapses to
\begin{equation}\label{e-W0-factored}
W_0=ac\,(bf-ag)\,(cf+df+ef+cg+dg),
\end{equation}
which vanishes precisely when $ag=bf$.

\subsubsection*{Vanishing of $W_1$}
Factoring $bc>0$ out of Eq.~\eqref{e-W1},
\begin{equation}\label{e-W1-step}
\frac{W_1}{bc}=(f+g)(adg+bcf+bef)-g\,(acf+acg+adf+aef+adg+bef).
\end{equation}
Expanding,
\begin{equation}\label{e-W1-expand}
\begin{aligned}
(f+g)(adg+bcf+bef)
&=adfg+bcf^2+bef^2+adg^2+bcfg+befg,\\[3pt]
g\,(acf+acg+adf+aef+adg+bef)
&=acfg+acg^2+adfg+aefg+adg^2+befg.
\end{aligned}
\end{equation}
The terms $adfg$, $adg^2$, and $befg$ cancel, leaving
\begin{equation}\label{e-W1-mid}
\frac{W_1}{bc}=\bigl(bcf^2+bcfg+bef^2\bigr)-\bigl(acfg+acg^2+aefg\bigr).
\end{equation}
Replacing $bf=ag$ maps the first group onto the second term by term,
\begin{equation}\label{e-W1-map}
bcf^2\to acfg,\qquad bcfg\to acg^2,\qquad bef^2\to aefg,
\end{equation}
so that $W_1=0$; equivalently,
\begin{equation}\label{e-W1-factored}
W_1=bc\,(bf-ag)\,(cf+cg+ef).
\end{equation}
With $W_0=0$ and $W_1=0$ we conclude $W(n)\equiv0$ for all $n$, as claimed.

\subsubsection*{Collapse of the switch}
Because $W(n)\equiv0$, we have $\partial K_{24}/\partial m=0$: the response is
independent of $m$ and may be evaluated at any convenient value. Taking
$m\to\infty$ in Eq.~\eqref{e-Kmn} retains only the leading coefficients,
\begin{equation}\label{e-Klimit}
K_{24}(m,n)=\frac{\mathcal{B}_+(2,n)}{\mathcal{B}_+(4,n)}
=\frac{bc\,(f+g)}{cg\,(a+nb)}
=\frac{b\,(f+g)}{g\,(a+nb)}.
\end{equation}
Imposing the protection condition $bf=ag$ (so that $f=ag/b$) gives
$b(f+g)=ag+bg=g(a+b)$, and therefore
\begin{equation}\label{e-Kab}
K_{24}(m,n)=\frac{a+b}{a+nb}.
\end{equation}
Finally, writing
\begin{equation}\label{e-theta}
\frac{a+nb}{a+b}=1+\frac{(n-1)\,b}{a+b}=1+(n-1)\,\theta_{\mathrm{eq}},
\qquad
\theta_{\mathrm{eq}}=\frac{b}{a+b}
=\frac{\ell_{\mathrm{eq}}(2\to3)}{\ell_{\mathrm{eq}}(2\to3)+\ell_{\mathrm{eq}}(2\to1)},
\end{equation}
we obtain the protocol-independent switch
\begin{equation}\label{e-Kfinal}
K_{24}(m,n)=K_{24}(n)=\frac{1}{1+(n-1)\,\theta_{\mathrm{eq}}},
\end{equation}
as quoted in the main text and plotted in Fig.~\ref{fig:Fig3}A.

\begin{table}[t]
\centering
\renewcommand{\arraystretch}{1.25}
\setlength{\tabcolsep}{10pt}
\begin{tabular}{c|c|c|c|c}
\toprule
\textbf{Tree} & $\boldsymbol{T_2}$ & $\boldsymbol{(\alpha,\beta)}$
              & $\boldsymbol{T_4}$ & $\boldsymbol{(\alpha,\beta)}$ \\
\midrule
$T_{\#1}$ & $2\to3\to4\to1$ & $(0,-)$ & $4\to1$          & $(0,0)$ \\
$T_{\#2}$ & $2\to1$          & $(0,0)$ & $4\to1$          & $(0,0)$ \\
$T_{\#3}$ & $2\to1$          & $(0,0)$ & $4\to1$          & $(0,0)$ \\
$T_{\#4}$ & $2\to1$          & $(0,0)$ & $4\to3\to2\to1$ & $(0,+)$ \\
$T_{\#5}$ & $2\to1$          & $(0,0)$ & $4\to1$          & $(0,0)$ \\
$T_{\#6}$ & $2\to1$          & $(0,0)$ & $4\to3\to1$      & $(+,0)$ \\
$T_{\#7}$ & $2\to3\to1$      & $(+,-)$ & $4\to1$          & $(0,0)$ \\
$T_{\#8}$ & $2\to3\to1$      & $(+,-)$ & $4\to3\to1$      & $(+,0)$ \\
\bottomrule
\end{tabular}
\caption{Tree-level data for the four-state network with energetic edges
$1\to3$ and $2\to3$. For every spanning tree rooted at vertex~$1$, the table
lists the unique simple paths $T_2$ and $T_4$ from vertices~$2$ and~$4$ to the
root, together with their $(\alpha,\beta)$ classes. These classes identify the
trees contributing to each $A_{\alpha,\beta}(i)$, and hence to the coefficients
$\mathcal{B}(i,n)$ that enter the derivation of $W(n)\equiv0$ in Sec.~S3.}
\label{t-1}
\end{table}

\begin{table}[t]
\centering
\renewcommand{\arraystretch}{1.25}
\begin{tabular}{@{}ll@{\hspace{3em}}ll@{}}
\toprule
Rate & Value & Rate & Value \\
\midrule
$\ell_{\mathrm{eq}}(2\to1)$ & $6.8910$              & $\ell_{\mathrm{eq}}(1\to2)$ & $2.6282\times10^{-3}$ \\
$\ell_{\mathrm{eq}}(3\to1)$ & $3.1539\times10^{-3}$ & $\ell_{\mathrm{eq}}(1\to3)$ & $0.0232$              \\
$\ell_{\mathrm{eq}}(4\to1)$ & $3.7021\times10^{-3}$ & $\ell_{\mathrm{eq}}(1\to4)$ & $67.2453$             \\
$\ell_{\mathrm{eq}}(3\to4)$ & $1.2147\times10^{4}$  & $\ell_{\mathrm{eq}}(4\to3)$ & $4.9210$              \\
$\ell_{\mathrm{eq}}(3\to2)$ & $0.4748$              & $\ell_{\mathrm{eq}}(2\to3)$ & $9.1599\times10^{3}$  \\
\bottomrule
\end{tabular}
\caption{Equilibrium transition rates $\ell_{\mathrm{eq}}(i\to j)$ for the
four-state network of Fig.~\ref{fig:Fig2} (thermodynamic switch). The rates
satisfy detailed balance, so the undriven system ($m=n=1$) sits at equilibrium,
and they obey the energetic-protection condition $w(T_{\#6})=w(T_{\#7})$, i.e.\
$\ell_{\mathrm{eq}}(2\to1)\,\ell_{\mathrm{eq}}(4\to3)
=\ell_{\mathrm{eq}}(2\to3)\,\ell_{\mathrm{eq}}(4\to1)$.}
\label{t-2}
\end{table}

\begin{table}[t]
\centering
\renewcommand{\arraystretch}{1.25}
\begin{tabular}{@{}ll@{\hspace{3em}}ll@{}}
\toprule
Rate & Value & Rate & Value \\
\midrule
$\ell_{\mathrm{eq}}(2\to1)$ & $35.9740$             & $\ell_{\mathrm{eq}}(1\to2)$ & $0.0040$  \\
$\ell_{\mathrm{eq}}(3\to1)$ & $89.4797$             & $\ell_{\mathrm{eq}}(1\to3)$ & $0.0016$  \\
$\ell_{\mathrm{eq}}(4\to1)$ & $1.8719\times10^{3}$  & $\ell_{\mathrm{eq}}(1\to4)$ & $0.0220$  \\
$\ell_{\mathrm{eq}}(3\to4)$ & $0.4914$              & $\ell_{\mathrm{eq}}(4\to3)$ & $0.7557$  \\
$\ell_{\mathrm{eq}}(3\to2)$ & $0.0895$              & $\ell_{\mathrm{eq}}(2\to3)$ & $0.0145$  \\
\bottomrule
\end{tabular}
\caption{Equilibrium transition rates $\ell_{\mathrm{eq}}(i\to j)$ used in
Fig.~\ref{fig:Fig3} (both panels). As in Table~\ref{t-2}, they satisfy detailed
balance and the protection condition $w(T_{\#6})=w(T_{\#7})$. The two panels of
Fig.~\ref{fig:Fig3} share these rates and differ only in the second energetic
edge ($2\to3$ in~A, $3\to4$ in~B), isolating the effect of that choice on the
protocol dependence of the switch.}
\label{t-3}
\end{table}

\begin{table}[t]
\centering
\renewcommand{\arraystretch}{1.25}
\begin{tabular}{@{}ll@{\hspace{3em}}ll@{}}
\toprule
Rate & Value & Rate & Value \\
\midrule
$\ell_{\mathrm{eq}}(2\to1)$ & $19.6867$   & $\ell_{\mathrm{eq}}(1\to2)$ & $0.4093$              \\
$\ell_{\mathrm{eq}}(3\to1)$ & $4.3067$    & $\ell_{\mathrm{eq}}(1\to3)$ & $924.5939$            \\
$\ell_{\mathrm{eq}}(4\to1)$ & $395.0976$  & $\ell_{\mathrm{eq}}(1\to4)$ & $398.0417$            \\
$\ell_{\mathrm{eq}}(3\to4)$ & $0.0048$    & $\ell_{\mathrm{eq}}(4\to3)$ & $1.0137$              \\
$\ell_{\mathrm{eq}}(3\to2)$ & $99.5666$   & $\ell_{\mathrm{eq}}(2\to3)$ & $1.0280\times10^{6}$  \\
\bottomrule
\end{tabular}
\caption{Equilibrium transition rates $\ell_{\mathrm{eq}}(i\to j)$ for the
coupled-edge example of Fig.~\ref{fig:Fig4}, in which the edges $1\to3$ and
$3\to4$ are driven together by a common factor~$m$. The rates satisfy detailed
balance, so the undriven system ($m=1$) sits at equilibrium; unlike
Tables~\ref{t-2} and~\ref{t-3}, they need not satisfy any protection condition.}
\label{t-4}
\end{table}
\end{document}